\documentclass{tlp}

\usepackage{amsmath,amssymb}
\usepackage{url}
\usepackage{color}
\usepackage{enumitem}
\usepackage[all]{xy}
\usepackage{algpseudocode}
\usepackage{svg}
\usepackage{fancyvrb}
\usepackage{hyperref}

\newcommand{\tuple}[1]{\langle #1 \rangle}

\def\At{\mathit{At}}

\newcommand\Lb[1]{\mathit{Lb}(#1)}
\newcommand\HH[1]{\mathit{H}(#1)}
\newcommand\Hd[1]{\mathit{Head}(#1)}
\newcommand\Bd[1]{\mathit{Body}(#1)}
\newcommand\BB[1]{\mathit{B}^+(#1)}
\newcommand\Bdp[1]{\mathit{Body}^+(#1)}
\newcommand\Bdn[1]{\mathit{Body}^-(#1)}
\def\SM{\mathit{SM}}
\def\JM{\mathit{JM}}
\def\SPM{\mathit{SPM}}
\newcommand{\Sup}[3]{\mathit{SUP}(#2,#1,#3)}

\newcommand{\xclingo}[0]{\texttt{xclingo}}

\newtheorem{theorem}{Theorem}
\newtheorem{proposition}{Proposition}

\newtheorem{corollary}{Corollary}
\newtheorem{definition}{Definition}
\newtheorem{example}{Example}
\def\qed{\hfill $\Box$}

\def\Not{\hbox{\rm \em not }}
\def\eqdef{\stackrel{\text{df}}{=}}

\begin{document}

\lefttitle{P. Cabalar and B. Mu\~niz}

\jnlPage{1}{8}
\jnlDoiYr{2021}
\doival{10.1017/xxxxx}

\title[Model Explanation via Support Graphs]{Model Explanation via Support Graphs
}

\begin{authgrp}
\author{\sn{Pedro} \gn{Cabalar}}
\affiliation{University of A Coru\~na, Spain}
\author{\sn{Brais} \gn{Mu{\~n}iz}}
\affiliation{University of A Coru\~na, Spain}
\end{authgrp}

\history{\sub{xx xx xxxx;} \rev{xx xx xxxx;} \acc{xx xx xxxx}}

\maketitle

\begin{abstract}
In this note, we introduce the notion of support graph to define explanations for any model of a logic program.
An explanation is an acyclic support graph that, for each true atom in the model,  induces a proof in terms of program rules represented by labels.
A classical model may have zero, one or several explanations: when it has at least one, it is called a justified model.
We prove that all stable models are justified whereas, in general, the opposite does not hold, at least for disjunctive programs.
We also provide a meta-programming encoding in Answer Set Programming that generates the explanations for a given stable model of some program.
We prove that the encoding is sound and complete, that is, there is a one-to-one correspondence between each answer set of the encoding and each explanation for the original stable model.
\end{abstract}




\section{Introduction}
\label{sec:intro}

In the past few years, Artificial Intelligence (AI) systems have made great advancements, generally at the cost of increasing their scale and complexity.
Although symbolic AI approaches have the advantage of being verifiable, the number and size of possible justifications generated to explain a given result may easily exceed the capacity of human comprehension.
Consider, for instance, the case of Answer Set Programming (ASP)~\cite{BET11}, a successful logic programming paradigm for practical Knowledge Representation and problem solving. 
Even for a positive program, whose answer set is unique, the number of proofs for an atom we can form using \emph{modus ponens} can be exponential.
It makes sense, then, to generate explanations through the typical ASP problem solving orientation.
Namely, we may consider each explanation   \emph{individually} as one solution to the ``explainability problem'' (that is, explaining a model) and let the user decide to generate one, several or all of them, or perhaps to impose additional preference conditions as done with optimisation problems in ASP.

In this technical note, we describe a formal characterisation of explanations in terms of graphs constructed with atoms and program rule labels.
Under this framework, models may be \emph{justified}, meaning that they have one or more \emph{explanations}, or \emph{unjustified} otherwise.
We prove that all stable models are justified whereas, in general, the opposite does not hold, at least for disjunctive programs.
We also provide an ASP encoding to generate the explanations of a given answer set of some original program, proving the soundness and completeness of this encoding.

The rest of this note is structured as follows.
Section~\ref{sec:graphs} contains the formal definitions for explanations and their properties with respect to stable models.
Section~\ref{sec:encoding} describes the ASP encoding and proves its soundness and completeness.
Section~\ref{sec:related-work} briefly comments on related work and, finally, Section~\ref{sec:conclusions} concludes the paper.

\section{Explanations as Support Graphs}
\label{sec:graphs}

We start from a finite\footnote{We leave the study of infinite signatures for future work. This will imply explanations of infinite size, but each one should contain a finite proof for each atom.} signature $\At$, a non-empty set of propositional atoms.
A \emph{(labelled) rule} is an implication of the form:
\begin{eqnarray}
\ell : p_1 \vee \dots \vee p_m \leftarrow q_1 \wedge \dots \wedge q_n \wedge \neg s_1 \wedge \dots \wedge \neg s_j \wedge \neg\neg t_1 \wedge \dots \wedge \neg\neg t_k \label{f:rule}
\end{eqnarray}
Given a rule $r$ like \eqref{f:rule}, we denote its label as $\Lb{r}\eqdef \ell$.
We also call the disjunction in the consequent $p_1 \vee \dots \vee p_m$ the \emph{head} of $r$, written $\Hd{r}$, and denote the set of head atoms as $\HH{r} \eqdef \{p_1, \dots, p_m\}$;
the conjunction in the antecedent is called the \emph{body} of $r$ and denoted as $\Bd{r}$.
We also define the positive and negative parts of the body respectively as the conjunctions $\Bdp{r} \eqdef q_1 \wedge \dots \wedge q_n$ and $\Bdn{r} \eqdef \neg s_1 \wedge \dots \wedge \neg s_j \wedge \neg\neg t_1 \wedge \dots \wedge \neg\neg t_k$.
The atoms in the positive body are represented as $\BB{r} \eqdef \{q_1,\dots,q_n\}$. 
As usual, an empty disjunction (resp. conjunction) stands for $\bot$ (resp. $\top$).
%
%
A rule $r$ with empty head $\HH{r}=\emptyset$ is called a \emph{constraint}.
On the other hand, when $\HH{r}=\{p\}$ is a singleton, $\BB{r}=\emptyset$ and $\Bdn{r}=\top$ the rule has the form $\ell: p \leftarrow \top$ and is said to be a \emph{fact}, simply written as $\ell:p$.
The use of double negation in the body allows representing elementary choice rules.
For instance, we will sometimes use the abbreviation $\ell : \{p\} \leftarrow B$ to stand for $\ell : p \leftarrow B \wedge \neg \neg p$.
A \emph{(labelled) logic program} $P$ is a set of labelled rules where no label is repeated.
Note that $P$ may still contain two rules $r, r'$ with same body and head $\Bd{r}=\Bd{r'}$ and $\HH{r}=\HH{r'}$, but different labels $\Lb{r} \neq \Lb{r'}$.
A program $P$ is \emph{positive} if $\Bdn{r}=\top$ for all rules $r \in P$.
A program $P$ is \emph{non-disjunctive} if $|\HH{r}|\leq 1$ for every rule $r \in P$.
Finally, $P$ is \emph{Horn} if it is both positive and non-disjunctive: note that this may include (positive) constraints $\bot \leftarrow B$.

A propositional interpretation $I$ is any subset of atoms $I \subseteq \At$.
We say that a propositional interpretation is a \emph{model} of a labelled program $P$ if $I \models \Bd{r} \to \Hd{r}$ in classical logic, for every rule $r \in P$.
The  \emph{reduct} of a labelled program $P$ with respect to $I$, written $P^I$, is a simple extension of the standard reduct by~\cite{GL88} that collects now the \emph{labelled} positive rules:
\begin{eqnarray*}
P^I \eqdef \{ \ \Lb{r}: \mathit{Head}(r) \leftarrow \Bdp{r} \ \mid \ r \in P, \ I \models \Bdn{r}\ \}
\end{eqnarray*}
As usual, an interpretation $I$ is a \emph{stable model} (or \emph{answer set}) of a program $P$ if $I$ is a minimal model of $P^I$.
Note that, for the definition of stable models, the rule labels are irrelevant.
We write $\SM(P)$ to stand for the set of stable models of $P$.
%

We define the rules of a program $P$ that \emph{support} an atom $p$ under interpretation $I$ as
$\Sup{I}{P}{p} \eqdef \{ r \in P \mid p \in \HH{r}, I \models \Bd{r}\}$
%
that is, rules with $p$ in the head whose body is true w.r.t. $I$.
The next proposition proves that, given $I$, the rules that support $p$ in the reduct $P^I$ are precisely the positive parts of the rules that support $p$ in $P$.
\begin{proposition}\label{prop:redsupport}
For any model $I \models P$ of a program $P$ and any atom $p \in I$:
$\Sup{I}{P^I}{p}=\Sup{I}{P}{p}^I$.
\end{proposition}
\begin{proof}
We prove first $\supseteq$: suppose $r \in \Sup{I}{P}{p}$ and let us call $r' = \Lb{r}: \Hd{r} \leftarrow \Bdp{r}$.
Then, by definition, $I \models  \Bd{r}$ and, in particular, $I \models \Bdn{r}$, so we conclude $r' \in P^I$.
To see that $r' \in \Sup{I}{P^I}{p}$, note that $I \models \Bd{r}$ implies $I \models \Bdp{r} = \Bd{r'}$.

For the $\subseteq$ direction, take any $r' \in \Sup{I}{P^I}{p}$.
By definition of reduct, we know that $r'$ is a positive rule and that there exists some $r \in P$ where $\Lb{r}=\Lb{r'}$, $\HH{r}=\HH{r'}$, $\BB{r}=\BB{r'}$ and $I \models \Bdn{r}$.
Consider any rule $r$ satisfying that condition (we could have more than one): we will prove that $r \in \Sup{I}{P}{p}$.
Since $r'\in \Sup{I}{P^I}{p}$, we get $I \models \Bd{r'}$ but this is equivalent to $I \models \Bdp{r}$.
However, as we had $I \models \Bdn{r}$, we conclude $I \models \Bd{r}$ and so $r$ is supported in $P$ given $I$.
\end{proof}

\begin{definition}[Support Graph/Explanation]\label{def:exp}
Let $P$ be a labelled program and $I$ a classical model of $P$. 
A \emph{support graph} $G$ of $I$ under $P$ is a labelled directed graph $G=\tuple{I,E,\lambda}$ whose vertices are the atoms in $I$, the edges in $E \subseteq I \times I$ connect pairs of atoms, the function $\lambda: I \to \Lb{P}$ assigns a label to each atom, and $G$ further satisfies:
\begin{enumerate}[label={\rm (\roman*)},leftmargin=20pt]
\item\label{def:exp.2} $\lambda$ is injective
\item\label{def:exp.3} for every $p \in I$, the rule $r$ such that $\Lb{r}=\lambda(p)$ satisfies:\\
$r \in \Sup{I}{P}{p}$ and $\BB{r} = \{ q \mid (q,p) \in E \}$.
\end{enumerate}
A support graph $G$ is said to be an \emph{explanation} if it additionally satisfies:
\begin{enumerate}[label={\rm (\roman*)},leftmargin=20pt]
\setcounter{enumi}{2}
\item\label{def:exp.1} $G$ is acyclic.\qed
\end{enumerate}
\end{definition}
Condition \ref{def:exp.2} means that there are no repeated labels in the graph, i.e., $\lambda(p) \neq \lambda(q)$ for different atoms $p, q \in I$.
Condition \ref{def:exp.3} requires that each atom $p$ in the graph is assigned the label $\ell$ of some rule with $p$ in the head, with a body satisfied by $I$ and whose atoms in the positive body form all the incoming edges for $p$ in the graph.
Intuitively, labelling $p$ with $\ell$ means that the corresponding (positive part of the) rule has been fired, ``producing'' $p$ as a result.
Since a label cannot be repeated in the graph, each rule can only be used to produce one atom, even though the rule head may contain more than one (when it is a disjunction).
It is not difficult to see that an explanation $G=\tuple{I,E,\lambda}$ for a model $I$ is uniquely determined by its atom labelling $\lambda$.
This is because condition \ref{def:exp.3} about $\lambda$ in Definition \ref{def:exp} uniquely specifies all the incoming edges for all the nodes in the graph.
On the other hand, of course, not every arbitrary atom labelling corresponds to a well-formed explanation.
We will sometimes abbreviate an explanation $G$ for a model $I$ by just using its labelling $\lambda$ represented as a set of pairs of the form $\lambda(p):p$ with $p \in I$.

\begin{definition}[Supported/Justified model]
A classical model $I$ of a labelled program $P$ if $I \models P$ is said to be a \emph{supported model} of $P$ if there exists some support graph of $I$ under $P$.
Moreover, $I$ is said to be a \emph{justified model} of $P$ if there exists some explanation $G$ (i.e. acyclic support graph) of $I$ under $P$.
We write $\SPM(P)$ and $\JM(P)$ to respectively stand for the set of supported and justified models of $P$.\qed
\end{definition}

Obviously all justified models are supported $\JM(P) \subseteq \SPM(P)$ but, in general, the opposite does not hold, as we will see later.
Our main focus, however, is on justified models, since we will relate them to proofs, that are always acyclic.
We can observe that not all models are justified, whereas a justified model may have more than one explanation, as we illustrate next.
\begin{example}\label{ex:1}
Consider the labelled logic program $P$
\begin{eqnarray*}
\ell_1: \ a \vee b 
\quad\quad\quad
\ell_2: \ d \leftarrow a \wedge \neg c
\quad\quad\quad
\ell_3: \ d \leftarrow \neg b
\end{eqnarray*}
No model $I \models P$ with $c \in I$ is justified since $c$ does not occur in any head, so its support is always empty $\Sup{I}{P}{c}=\emptyset$ and $c$ cannot be labelled.
The models of $P$ without $c$ are $\{b\}$, $\{a,d\}$, $\{b,d\}$ and $\{a,b,d\}$ but only the first two are justified.
The explanation for $I=\{b\}$ corresponds to the labelling $\{(\ell_1: b)\}$ (it forms a graph with a single node).
Model $I=\{a,d\}$ has the two possible explanations:
\begin{eqnarray}
\ell_1:a \longrightarrow \ell_2:d
\hspace{100pt}
\ell_1:a \quad\quad \ell_3:d
\label{f:twoexplanations}
\end{eqnarray}
Model $I=\{b,d\}$ is not justified: we have no support for $d$ given $I$, $\Sup{I}{P}{d}=\emptyset$, because $I$ satisfies neither bodies of $\ell_2$ nor $\ell_3$.
On the other hand, model $\{a,b,d\}$ is not justified either, because  $\Sup{I}{P}{a}=\Sup{I}{P}{b}=\{\ell_1\}$ and we cannot use the same label $\ell_1$ for two different atoms $a$ and $b$ in a same explanation (condition~\ref{def:exp.2} in Def.~\ref{def:exp}).\qed
\end{example}

\begin{definition}[Proof of an atom]\label{def:proof}
Let $I$ be a model of a labelled program $P$, $G=\tuple{I,E,\lambda}$ an explanation for $I$ under $P$ and let $p \in I$.
The \emph{proof} for $p$ induced by $G$, written $\pi_G(p)$, is the derivation: 
\begin{eqnarray*}
\pi_G(p) & \eqdef & \frac{\pi_G(q_1) \ \dotsc \ \pi_G(q_n)}{p} \  \lambda(p),
\end{eqnarray*}
\noindent where, if $r \in P$ is the rule satisfying $\Lb{r}=\lambda(p)$, then $\{q_1,\dots,q_n\}= \BB{r}$.
%
When $n=0$, the  derivation antecedent $\pi_G(q_1) \ \dotsc \ \pi_G(q_n)$ is replaced by $\top$ (corresponding to the empty conjunction).\qed
\end{definition}

\begin{example}\label{ex:proof}
Let $P$ be the labelled logic program:
\begin{eqnarray*}
\ell_1: \ p 
\quad\quad\quad
\ell_2: \ q \leftarrow p
\quad\quad\quad
\ell_3: \ r \leftarrow p, q 
\end{eqnarray*}
$P$ has a unique justified model $\{p,q,r\}$ whose explanation is shown in Figure~\ref{fig:proof} (left) whereas the induced proof for atom $r$ is shown in Figure~\ref{fig:proof} (right).
\qed
\end{example}

\begin{figure}[htbp]
\[
\begin{array}{c@{\hspace{40pt}}c}
\xymatrix @-4mm {
\ell_1: \ p  \ar[r] \ar@/^{15pt}/[rr] & \ell_2: \ q \ar[r] & \ell_3: \ r
}
&
\cfrac {
  \cfrac {
    \cfrac{\ \top\ }{p}\ (\ell_1)
  }
  {q}\ (\ell_2)
  \quad\quad
  \cfrac{\ \top\ }{p}\ (\ell_1)
}
{r} \ (\ell_3)\\[10pt]
\text{Explanation} & \text{Proof for atom $r$}
\end{array}
\]
\caption{Some results for model $\{p,q,r\}$ of program in Example~\ref{ex:proof}.}
\label{fig:proof}
\end{figure}
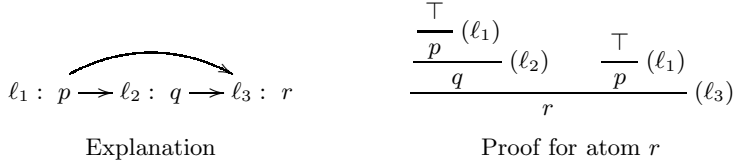
The next proposition trivially follows from the definition of explanations:

\begin{proposition}\label{prop:MP}
If $P$ is a Horn program, and $G$ is an explanation for a model $I$ of $P$ then, for every atom, $p \in I$, $\pi_G(p)$ corresponds to a Modus Ponens derivation of $p$ using the rules in $P$.
\end{proposition}

It is worth mentioning that explanations do not generate any arbitrary Modus Ponens derivation of an atom, but only those that are globally ``coherent'' in the sense that, if any atom $p$ is repeated in a proof, it is always justified repeating \emph{the same subproof}.


%
%

In the previous examples, justified and stable models coincided: one may wonder whether this is a general property.
As we see next, however, every stable model is justified but, in general, the opposite may not hold.
To prove that stable models are justified, we start proving a correspondence between explanations for any model $I$ of $P$ and explanations under $P^I$.

\begin{proposition}\label{prop:reduct}
Let $I$ be a model of program $P$. 
Then $G$ is an explanation for $I$ under $P$ iff $G$ is an explanation for $I$ under $P^I$.
\end{proposition}
\begin{proof}
By Proposition~\ref{prop:redsupport}, for any atom $p \in I$, the labels in $\Sup{I}{P}{p}$ and $\Sup{I}{P^I}{p}$ coincide, so there is no difference in the ways in which we can label $p$ in explanations for $P$ and for $P^I$.
On the other hand, the rules in $\Sup{I}{P^I}{p}$ are the positive parts of the rules in $\Sup{I}{P}{p}$, so the graphs we can form are also the same.
\end{proof}


\begin{corollary}
$I \in \JM(P)$ iff $I \in \JM(P^I)$.
\end{corollary}

\begin{theorem}\label{th:stable-are-justified}
Stable models are justified: $\SM(P) \subseteq \JM(P)$.
\end{theorem}
\begin{proof}
Let $I$ be a stable model of $P$.
To prove that there is an explanation $G$ for $I$ under $P$, we can use Proposition~\ref{prop:redsupport} and just prove that there is some explanation $G$ for $I$ under $P^I$.
We will build the explanation with a non-deterministic algorithm where, in each step $i$, we denote the graph $G_i$ as $G_i=\tuple{I_i,E_i,\lambda_i}$ and represent the labelling $\lambda_i$ as a set of pairs of the form $(\ell:p)$ meaning $\ell=\lambda(p)$.
The algorithm proceeds as follows:
\begin{algorithmic}[1]
\State $I_0 \gets \emptyset; E_0  \gets \emptyset; \lambda_0 \gets \emptyset$ 
\State $G_0 = \tuple{I_0,E_0,\lambda_0}$ 
\State $i \gets 0$ 
\While{$I_i \not\models P^I$ }
  \State{Pick a rule $r \in P^I$ s.t. $I_i \models \Bd{r} \wedge \neg \Hd{r}$}
  \State{Pick an atom $p \in I \cap \HH{r}$}
  \State{$I_{i+1} \gets I_i \cup \{p\} $}
  \State{$\lambda_{i+1} \gets \lambda_i \cup \{(\ell: p)\} $}
  \State{$E_{i+1} \gets E_i \cup \{(q,p) \mid q \in \BB{r}\} $}
  \State{$G_i \gets \tuple{I_i,E_i,\lambda_i}$}
  \State $i \gets i+1$
\EndWhile
\end{algorithmic}

\noindent The existence of a rule $r \in P^I$ in line 5 is guaranteed because the {\bf while} condition asserts $I_i \not\models P^I$ and so there must be some rule whose positive body is satisfied by $I_i$ but its head is not satisfied.
We prove next that the existence of an atom $p \in I \cap \Hd{r}$ (line 5) is also guaranteed.
First, note that the {\bf while} loop maintains the invariant $I_i \subseteq I$, since $I_0=\emptyset$ and $I_i$ only grows with atoms $p$ (line 7) that belong to $I$ (line 6).
Therefore, $I_i \models \Bd{r}$ implies $I \models \Bd{r}$, but since $I \models P^I$, we also conclude $I \models r$ and thus $I \models \Hd{r}$ that is $I \cap \HH{r} \neq \emptyset$, so we can always pick some atom $p$ in that intersection.
Now, note that the algorithm stops because, in each iteration, $I_i$ grows with exactly one atom from $I$ that was not included before, since $I_i \models \neg \Hd{r}$, and so, this process will stop provided that $I$ is finite.
The {\bf while} stops satisfying $I_i \models P^I$ for some value $i=n$.
Moreover, $I_n = I$, because otherwise, as $I_i \subseteq I$ is an invariant, we would conclude $I_n \subset I$ and so $I$ would not be a minimal model of $P^I$, which contradicts that $I$ is a stable model of $P$.
We remain to prove that the final $G_n=\tuple{I_n,E_n,\lambda_n}$ is a correct explanation for $I$ under $P^I$.
As we said, the atoms in $I$ are the graph nodes $I_n=I$.
Second, we can easily see that $G_n$ is acyclic because each iteration adds a new node $p$ and links this node to previous atoms from $\BB{r} \subseteq I_i$ (remember $I_i \models \Bd{r}$) so no loop can be formed.
Third, no rule label can be repeated, because we go always picking a rule $r$ that is new, since it was not satisfied in $I_i$ but becomes satisfied in $I_{i+1}$ (the rule head $\Hd{r}$ becomes true).
Last, for every $p \in I$, it is not hard to see that the (positive) rule $r \in P^I$ such that $\Lb{r}=\lambda_n(p)$ satisfies $p \in \HH{r}$ and $\BB{r} = \{ q \mid (q,p) \in E \}$ by the way in which we picked $r$ and inserted $p$ in $I_i$, whereas $I \models \Bd{r}$ because $I_i \models \Bd{r}$, $r$ is a positive rule and $I_i \subseteq I$.
\end{proof}

As a result, we get $\SM(P) \subseteq \JM(P) \subseteq \SPM(P)$, that is, justified models lay in between stable and supported.

\begin{proposition}\label{prop:least}
If $P$ is a consistent Horn program then it has a unique justified model $I$ that coincides with the least model of $P$.
\end{proposition}
\begin{proof}
Since $P$ is Horn and consistent (all constraints are satisfied) its unique stable model is the least model $I$. 
By Theorem~\ref{th:stable-are-justified}, $I$ is also justified by some explanation $G$.
We remain to prove that $I$ is the unique justified model.
Suppose there is another model $J \supset I$ (remember $I$ is the least model) justified by an explanation $G$ and take some atom $p \in J \setminus I$.
Then, by Proposition~\ref{prop:MP}, the proof for $p$ induced by $G$, $\pi_G(p)$, is a Modus Ponens derivation of $p$ using the rules in $P$.
Since Modus Ponens is sound and the derivation starts from facts in the program, this means that $p$ must be satisfied by any model of $P$, so $p \in I$ and we reach a contradiction.
\end{proof}

In general, the number of explanations for a single justified model can be exponential, even when the program is Horn, and so, has a unique justified and stable model corresponding to the least classical model, as we just proved.
As an example\footnote{This example was already introduced as Program 7.1 in~\cite{Fandinno15}.}:

\begin{example}[A chain of firing squads]
Consider the following variation of the classical \emph{Firing Squad Scenario} introduced by~\cite{Pearl99} for causal counterfactuals (although we do not use it for that purpose here).
We have an army distributed in $n$ squads of three soldiers each, a captain and two riflemen for each squad.
We place the squads in a sequence of $n$ consecutive hills $i=0,\dots,n-1$.
An unfortunate prisoner is at the last hill $n-1$, and is being aimed at by the last two riflemen.
At each hill $i$, the two riflemen $a_i$ and $b_i$ will fire if their captain $c_i$ gives a signal to fire.
But then, captain $c_{i+1}$ will give a signal to fire if she hears a shot from the previous hill $i$ in the distance.
Suppose captain $c_0$ gives a signal to fire.
Our logic program would have the form:
\[\begin{array}{r@{\ }l@{\quad\quad}r@{\ }l@{\quad\quad}r@{\ }l}
s_0: & \mathit{signal}_0
&
a_i: & \mathit{fireA}_i \leftarrow \mathit{signal}_i 
&
a'_{i+1}: & \mathit{signal}_{i+1} \leftarrow \mathit{fireA}_i \\
& & 
b_i: & \mathit{fireB}_i \leftarrow \mathit{signal}_i
&
b'_{i+1}: & \mathit{signal}_{i+1} \leftarrow \mathit{fireB}_i 
\end{array}
\]
for all $i=0,\dots,n-1$ where we assume (for simplicity) that $\mathit{signal}_{n}$ represents the death of the prisoner.
This program has one stable model (the least model) making true the $3 n + 1$ atoms occurring in the program. 
However, this last model has $2^n$ explanations because to derive $\mathit{signal}_{i+1}$ from level $i$, we can choose between any of the two rules $a'_i$ or $b'_i$ (corresponding to the two riflemen) in each explanation.
\qed
\end{example}

In many disjunctive programs, justified and stable models coincide.
For instance, the following example is an illustration of a program with disjunction and head cycles.
\begin{example}
Let $P$ be the program:
\begin{eqnarray*}
\ell_1: p \vee q 
\hspace{50pt}
\ell_2: q \leftarrow p
\hspace{50pt}
\ell_3: p \leftarrow q 
\end{eqnarray*}
This program has one justified model $\{p,q\}$ that coincides with the unique stable model and has two possible explanations, $\{(\ell_1:p),(\ell_2:q)\}$ and $\{(\ell_1:q),(\ell_3:p)\}$.\qed
\end{example}
However, in the general case, not every justified model is a stable model: we provide next a simple counterexample.
Consider the program $P$:
\begin{eqnarray*}
\ell_1: a \vee b 
\hspace{50pt}
\ell_2: a \vee c
\end{eqnarray*}
whose classical models are the five interpretations: $\{a\}$, $\{a,c\}$, $\{a,b\}$, $\{b,c\}$ and $\{a,b,c\}$.
The last one $\{a,b,c\}$ is not justified, since we would need three different labels and we only have two rules.
Each model $\{a,c\}$, $\{a,b\}$, $\{b,c\}$ has a unique explanation corresponding to the atom labellings $\{(\ell_1:a),(\ell_2:c)\}$, $\{(\ell_1:b),(\ell_2:a)\}$ and $\{(\ell_1:b),(\ell_2:c)\}$, respectively.
On the other hand, model $\{a\}$ has two possible explanations, corresponding to  $\{(\ell_1:a)\}$ and $\{(\ell_2:a)\}$.
Notice that, in the definition of explanation, there is no need to fire every rule with a true body in $I$ -- we are only forced to explain every true atom in $I$.
Note also that only the justified models $\{a\}$ and $\{b,c\}$ are also stable: this is due to the minimality condition imposed by stable models on positive programs, getting rid of the other two justified models $\{a,b\}$ and $\{a,c\}$.
The following theorem asserts that, for non-disjunctive programs, every justified model is also stable.

\begin{theorem}\label{th:nondisjunctive}
If $P$ is a non-disjunctive program, then $\SM(P)=\JM(P)$.\qed
\end{theorem}
\begin{proof}
Given Theorem~\ref{th:stable-are-justified}, we must only prove that, for non-disjunctive programs, every justified model is also stable.
Let $I$ be a justified model of $P$.
By Proposition~\ref{prop:reduct}, we also know that $I$ is a justified model of $P^{I}$. 
$P^{I}$ is a positive program and is non-disjunctive (since $P$ was non-disjunctive) and so, $P$ is a Horn program.
By Proposition~\ref{prop:least}, we know $I$ is also the \emph{least model} of $P^{I}$,
which makes it a stable model of $P$.
\end{proof}

Moreover, for non-disjunctive programs, we can prove that our definition of supported model, coincides with the traditional one in terms of fixpoints of the immediate consequences operator~\cite{vEm76} or as models of completion~\cite{Cla78}.
Given a non-disjunctive program $P$, let $T_P(I)$ be defined as $\{p \mid r \in P, I \models \Bd{r}, \Hd{r}=p\}$.
\begin{theorem}\label{th:nondisjunctive2}
If $P$ is a non-disjunctive program, then $I=T_P(I)$ iff $I \in \SPM(P)$.\qed
\end{theorem}
\begin{proof}
For left to right, suppose $I=T_P(I)$.
It is easy to see that this implies $I \models P$.
By definition of $T_P$, for each atom $p$ there exists some rule $r$ with $\Hd{r}=p$ and $I \models \Bd{r}$.
Let us arbitrarily pick one of those rules $r_p$ for each $p$.
Then we can easily form a support graph where $\lambda(p)=\Lb{r_p}$ and assign all the incoming edges for $p$ as $(q,p)$ such that $q \in \Bdp{r_p}$.

For right to left, suppose $I \models P$ and there is some support graph $G$ of $I$ under $P$.
We prove both inclusion directions for $I=T_P(I)$.
For $\subseteq$, suppose $p \in I$.
Then $p$ is a node in $G$ and there is a rule $r$ such that $\lambda(p)=\Lb{r}$, $p=\Hd{r}$ ($P$ is non--disjunctive) and $I \models \Bd{r}$.
But then $p \in T_P(I)$.
For $\supseteq$, take any $p \in T_P(I)$ and suppose $p \not\in I$.
Then, we have at least some rule $r \in P$ with $I \models \Bd{r}$ and $I \not\models \Hd{r} (=p)$, something that contradicts $I \models P$.
\end{proof}

To illustrate supported models in the disjunctive case, consider the program:
\begin{eqnarray*}
\ell_1: a \vee b \leftarrow c
\hspace{50pt}
\ell_2: c \leftarrow b
\end{eqnarray*}
The only justified model of this program is $\emptyset$ which is also stable and supported.
Yet, we also obtain a second supported model $\{b,c\}$ that is justified by the (cyclic) support graph with labelling $\{\ell_1:b, \ell_2:c\}$.

\section{An ASP encoding to compute explanations}
\label{sec:encoding}

In this section, we focus on the computation of explanations for a given stable model.
We assume that we use an ASP solver to obtain the answer sets of some program $P$ and that we have some way to label the rules.
For instance, we may use the code line number (or another tag specified by the user), followed by the free variables in the rule and some separator.
In that way, after grounding, we get a unique identifier for each ground rule.

To explain the answer sets of $P$ we may build the following (non-ground) ASP program $x(P)$ that can be fed with the (reified) true atoms in $I$ to build the ground program $x(P,I)$.
As we will prove, the answer sets of $x(P,I)$ are in one-to-one correspondence with the explanations of $I$.
The advantage of this technique is that, rather than collecting all possible explanations in a single shot, something that results too costly for explaining large programs, we can perform regular calls to an ASP solver for $x(P,I)$ to compute one, several or all explanations of $I$ on demand.
Besides, this provides a more declarative approach that can be easily extended to cover new features (such as, for instance, minimisation among explanations).

For each rule in $P$ of the form \eqref{f:rule}, $x(P)$ contains the set of rules:
\begin{eqnarray}
sup(\ell) & \leftarrow &  as(q_1) \wedge \dots \wedge as(q_n) \wedge as(p_i) \wedge \neg as(s_1) \wedge \dots \wedge \neg as(s_j)\\
& & \ \wedge \ \neg\neg as(t_1) \wedge \dots \wedge \neg\neg as(t_k) \label{f:xP.1}\\
\{ f(\ell,p_i) \} & \leftarrow & f(q_1) \wedge \dots \wedge f(q_n) \wedge as(p_i) \wedge sup(\ell)\label{f:xP.2}\\
\bot & \leftarrow & f(\ell,p_i) \wedge f(\ell,p_h) \label{f:xP.3}
\end{eqnarray}
for all $i,h=1\dots m$ and $i \neq h$, and, additionally $x(P)$ contains the rules:
\begin{eqnarray}
f(A) & \leftarrow & f(L,A) \wedge as(A) \label{f:xP.4}\\
\bot & \leftarrow & \Not f(A) \wedge as(A) \label{f:xP.5}\\
\bot & \leftarrow & f(L,A) \wedge f(L',A) \wedge L \neq L' \wedge as(A) \label{f:xP.6}
\end{eqnarray}
As we can see, $x(P)$ reifies atoms in $P$ using three predicates: $as(A)$ which means that atom $A$ is in the answer set $I$, so it is an initial assumption; $f(L,A)$ means that rule with label $L$ has been ``fired'' for atom $A$, that is, $\lambda(A)=L$; and, finally, $f(A)$ that just means that there exists some fired rule for $A$ or, in other words, we were able to derive $A$.
Predicate $sup(\ell)$ tells us that the body of the rule $r$ with label $\ell$ is ``supported'' by $I$, that is, $I \models \Bd{r}$.
Given any answer set $I$ of $P$, we define the program $x(P,I) \eqdef x(P) \cup \{as(A) \mid A \in I\}$.
It is easy to see that $x(P,I)$ becomes equivalent to the ground program containing the following rules:
\begin{alignat}{3}
\{ f(\ell,p) \} & \leftarrow f(q_1) \wedge \dots \wedge f(q_n) & \quad \quad \text{for each rule $r \in P$ like \eqref{f:rule},}\nonumber\\
& & \quad \quad I\models \Bd{r}, p \in \HH{r} \cap I  \label{f:xPI.1}\\
\bot & \leftarrow f(\ell,p_i) \wedge f(\ell,p_j) & \quad \quad \text{for each rule $r \in P$ like \eqref{f:rule},} \nonumber \\
& & \quad \quad p_i, p_j \in \HH{r}, \ p_i \neq p_j \label{f:xPI.2}\\
f(a) & \leftarrow f(\ell,a) & \quad \quad \text{for each $a \in I$} \label{f:xPI.3}\\
\bot & \leftarrow \Not f(a) & \quad \quad \text{for each $a \in I$}  \label{f:xPI.4}\\
\bot & \leftarrow f(\ell,a) \wedge f(\ell',a) &  \quad \quad \text{for each } a \in I, \ \ell \neq \ell' \label{f:xPI.5}
\end{alignat}

\begin{theorem}[Soundness]\label{th:sound}
Let $I$ be an answer set of $P$. For every answer set $J$ of program $x(P,I)$ there exists an explanation $G=\tuple{I,E,\lambda}$ of $I$ under $P$ such that $\lambda(a)=\ell$ iff $f(\ell,a) \in J$.\qed
\end{theorem}
\begin{proof}
We have to prove that $J$ induces a valid explanation $G$.
Let us denote $\At(J) \eqdef \{a \in \At \mid f(a) \in J\}$.
Since \eqref{f:xPI.3} is the only rule for $f(a)$, we can apply completion to conclude that $f(a) \in J$ iff $f(\ell,a) \in J$ for some label $\ell$.
So, the set $\At(J)$ contains the set of atoms for which $J$ assigns some label: we will prove that this set coincides with $I$.
We may observe that $I \subseteq \At(J)$ because for any $a \in I$ we have the constraint \eqref{f:xPI.4} forcing $f(a) \in J$.
On the other hand, $\At(J) \subseteq I$ because the only rules with $f(a)$ in the head are \eqref{f:xPI.3} and these are only defined for atoms $a \in I$.
To sum up, in any answer set $J$ of $x(P,I)$, we derive exactly the original atoms in $I$, $\At(J)=I$ and so, the graph induced by $J$ has exactly one node per atom in $I$.

Constraint \eqref{f:xPI.5} guarantees that atoms $f(\ell,a)$ have a functional nature, that is, we never get two different labels for a same atom $a$.
This allows defining the labelling function $\lambda(a)=\ell$ iff $f(\ell,a) \in J$.
We remain to prove that conditions \ref{def:exp.2}-\ref{def:exp.1} in Definition~\ref{def:exp} hold.
Condition \ref{def:exp.2} requires that $\lambda$ is injective, something guaranteed by~\eqref{f:xPI.2}.
Condition \ref{def:exp.3} requires that, informally speaking, the labelling for each atom $a$ corresponds to an activated, supported rule for $a$.
That is, if $\lambda(a)=\ell$, or equivalently $f(\ell,a)$, we should be able to build am edge $(q,a)$ for each atom in the positive body of $\ell$ so that atoms $q$ are among the graph nodes.
This is guaranteed by that fact that rule \eqref{f:xPI.1} is the only one with predicate $f(\ell,a)$ in the head.
So, if that ground atom is in $J$, it is because $f(q_i)$ are also in $J$ i.e. $q_i \in I$, for all atoms in the positive body of rule labelled with $\ell$.
Note also that \eqref{f:xPI.1} is such that $I \models \Bd{r}$, so the rule supports atom $p$ under $I$, that is, $r \in \Sup{I}{P}{p}$.
Let us call $E$ to the set of edges formed in this way.
Condition~\ref{def:exp.1} requires that the set $E$ of edges forms an acyclic graph.
To prove this last condition, consider the reduct program $x(P,I)^J$.
The only difference of this program with respect to $x(P,I)$ is that rules \eqref{f:xPI.1} have now the form:
\begin{eqnarray}
f(\ell,p) \leftarrow f(q_1) \wedge \dots \wedge f(q_n) \label{f:xPI.1b}
\end{eqnarray}
for each rule $r \in P$ like \eqref{f:rule}, $I\models \Bd{r}$, $p \in \HH{r} \cap I$ as before, but additionally $f(\ell,p) \in J$ so the rule is kept in the reduct.
Yet, the last condition is irrelevant since $f(\ell,p) \in J$ implies $f(p) \in J$ so $p \in \At(J) = I$.
Thus, we have exactly one rule \eqref{f:xPI.1b} in $x(P,I)^J$ per each choice \eqref{f:xPI.1} in $x(P,I)$.
Now, since $J$ is an answer set of $x(P,I)$, by monotonicity of constraints, it \eqref{f:xPI.2}, \eqref{f:xPI.4} and \eqref{f:xPI.5} and is an answer set of the rest of the program $P'$ formed by rules \eqref{f:xPI.1b} and \eqref{f:xPI.2}.
This means that $J$ is a minimal model of $P'$.
Suppose we have a cycle in $E$, formed by the (labelled) nodes and edges $(\ell_1:p_1) \longrightarrow \dots \longrightarrow (\ell_n:p_n) \longrightarrow (\ell_1:p_1)$.%
Take the interpretation $J'=J\setminus \{f(\ell_1,p_1), \dots, f(\ell_n,p_n), f(p_1), \dots, f(p_n)\}$.
Since $J$ is a minimal for $P'$ there must be some rule \eqref{f:xPI.1b} or \eqref{f:xPI.2} not satisfied by $J'$.
Suppose $J'$ does not satisfy some rule \eqref{f:xPI.2} so that $f(a) \not\in J'$ but $f(\ell,a) \in J' \subseteq J$.
This means we had $f(a) \in J$ since the rule was satisfied by $J$ so $a$ is one of the removed atoms $p_i$ belonging to the cycle.
But then $f(\ell,a)$ should have been removed $f(\ell,a)\not\in J'$ and we reach a contradiction.
Suppose instead that $J'$ does not satisfy some rule \eqref{f:xPI.1b}, that is, $f(\ell,p) \not\in J'$ and $\{f(q_1),\dots,f(g_n)\} \subseteq J' \subseteq J$.
Again, since the body holds in $J$, we get $f(\ell,p) \in J$ and so, $f(\ell,p)$ is one of the atoms in the cycle we removed from $J'$.
Yet, since $(\ell:p)$ is in the cycle, there is some incoming edge from some atom in the cycle and, due to the way in which atom labelling is done, this means that this edge must come from some atom $q_i$ with $1\leq i\leq n$ in the positive body of the rule whose label is $\ell$.
But, since this atom is in the cycle, this also means that $f(q_i) \not\in J'$ and we reach a contradiction.
\end{proof}

\begin{theorem}[Completeness]\label{th:complete}
Let $I$ be an answer set of $P$. For every explanation $G=\tuple{I,E,\lambda}$ of $I$ under $P$ there exists some answer set $J$ of program $x(P,I)$ where $f(\ell,a) \in J$ iff $\lambda(a)=\ell$ in $G$.\qed
\end{theorem}
\begin{proof}
Take $I$ an answer set of $P$ and $G=\tuple{I,E,\lambda}$ some explanation for $I$ under $P$ and let us define the interpretation:
\begin{eqnarray*}
J := \{f(a) \mid a \in I\} \cup \{f(\ell,a) \mid \lambda(a)=\ell \}
\end{eqnarray*}
We will prove that $J$ is an answer set of $x(P,I)$ or, in other words, that $J$ is a minimal model of $x(P,I)^J$.
First, we will note that $J$ satisfies $x(P,I)^J$ rule by rule.
For the constraints, $J$ obviously satisfy \eqref{f:xP.4} because it contains an atom $f(a)$ for each $a \in I$.
We can also see that $J$ satisfies \eqref{f:xPI.2} because graph $G$ does not contain repeated labels, so we cannot have two different atoms with the same label.
The third constraint \eqref{f:xPI.5} is also satisfied by $J$ because atoms $f(\ell,a), f(\ell',a)$ are obtained from $\lambda(a)$ that is a function that cannot assign two different labels to a same atom $a$.
Satisfaction of \eqref{f:xPI.2} is guaranteed since the head of this rule $f(a)$ is always some atom $a \in I$ and therefore $f(a) \in J$.
For the remaining rule, \eqref{f:xPI.1}, we have two cases.
If $f(\ell,p) \not\in J$ then the rule is not included in the reduct and so there is no need to be satisfied.
Otherwise, if $f(\ell,p) \in J$ then the rule in the reduct corresponds to \eqref{f:xPI.1b} and is trivially satisfied by $J$ because its only head atom holds in that interpretation.
Finally, to prove that $J$ is a minimal model of $x(P,I)^J$, take the derivation tree $\pi_G(a)$ for each atom $a \in I$.
Now, construct a new tree $\pi$ where we replace each atom $p$ in $\pi_G(a)$ by an additional derivation from $f(\ell,p)$ to $f(p)$ through rule \eqref{f:xPI.3}.
It is easy to see that $\pi$ constitutes a Modus Ponens proof for $f(a)$ under the Horn program $x(P,I)^J$ and the same reasoning can be applied to atom $f(\ell,a) \in J$ that is derived in the tree $\pi$ for $f(a)$.
Therefore, all atoms in $J$ must be included in any model of $x(P,I)^J$.
\end{proof}

\section{Related work}
\label{sec:related-work}

The current approach constitutes the formal basis of the new version of the explanation tool {\tt xclingo}~\cite{CabMun23} which also uses the ASP encoding from Section~\ref{sec:encoding} to compute the explanations.
Theorems~\ref{th:sound} and \ref{th:complete} prove, in this way, that the tool is sound and complete with respect to the definition of explanation provided in the current paper.

There exist many other approaches for explanation and debugging in ASP (see the survey~\cite{FandinnoS19}).
The closest approach to the current work is clearly the one based on \emph{causal graphs}~\cite{CabFanFin2014}.
%
Although we conjecture that a formal relation can be established (we plan this for future work), the main difference is that causal graphs are ``atom oriented'' whereas the current approach is model oriented.
%
For instance, in the firing squads example, the causal-graph explanation for the derivations of atoms $\mathit{signal}_4$ and $\mathit{signal}_8$ would contain algebraic expressions with \emph{all} the possible derivations for each one of those atoms.
In the current approach, however, we would get an individual derivation in each case, but additionally, the proof we get for $\mathit{signal}_4$ has to be \emph{the same one} we use for that atom inside the derivation of $\mathit{signal}_8$.


Justifications based on the positive part of the program were also used before in~\cite{esra13}.
There, the authors implemented an ad-hoc approach to the problem of solving biomedical queries, rather than a general ASP explanation tool.
%
%
%
%
%

%
Other examples of general approaches are the \emph{formal theory of justifications}~\cite{DeneckerBS15}, \emph{off-line justifications}~\cite{PonSon2008}, LABAS~\cite{ST16} (based on argumentation theory~\cite{Bond00, Dung09}) or s(CASP)~\cite{ArCarr20}.
All of them provide graph or tree-based explanations for an atom to be (or not) in a given answer set.
The formal theory of justifications was also extended to deal with nested graph based justifications~\cite{Marynissen22} and is actually a more general framework that allows covering other logic programming semantics.
System {\tt xASP}~\cite{xASP22} generates explanation graphs from~\cite{PonSon2008} and also uses an ASP meta-programming encoding.
In the case of s(CASP), it
%
proceeds in a top-down manner, building the explanation as an ordered list of literals extracted from the goal-driven satisfaction of the query.
%
%
An important difference with respect to this last group of approaches is that their explanations consider dependences through default negation.
To illustrate the effect, take the program:
\begin{eqnarray*}
\ell_1: switch & \\
\ell_2: light & \leftarrow & switch, \Not ab \\
\ell_3: ab & \leftarrow & blown\_fuse \\
\ell_4: ab & \leftarrow & broken\_bulb \\
\ell_5: ab & \leftarrow & blackout, \Not generator
\end{eqnarray*}
The only stable model is $\{switch,light\}$ and its unique explanation is the support graph 
\begin{eqnarray*}
\ell_1: switch \longrightarrow \ell_2:light
\end{eqnarray*}that is, the light is on because we toggled the switch.
Adding negative information would lead us to explain $\Not ab$ and obtain two explanations: one in which we also add that there is no blown fuse, no broken bulb and no blackout; the second one is similar, but instead of no blackout, we have a doubly negative dependence on generator: i.e. nothing prevents having a generator, even though we do not have it. 
Note how these explanations may easily get complicated: we could have to negate multiple alternative ways of breaking the bulb, even when \emph{none of them have happened}\footnote{We face here, somehow, a kind of qualification problem in the explanations.}.
Our approach consists, instead, in explaining the information that currently holds, assuming that other states of affairs will arise in terms of other alternative answer sets.
In other words, we refrain from using facts for which we have no evidence or reason to believe in our current model.

Another distinctive feature of our approach is that it provides explanations for disjunctive programs and, moreover, it has also allowed us to define supported and justified models for that case.
In fact, we plan to study potential connections between justified models and other approaches for disjunction not based in minimal models such as~\cite{ACPPV17} or~\cite{SE19}.

Other ASP explanation approaches have to do with comparing stable models or explaining their non-existence.
For instance, \cite{GebserPST08} uses a meta-programming technique to explain why a given model \emph{is not} an answer set of a given program.
More recently, ~\cite{EiSar19} considered the explanation of ASP programs that have no answer sets in terms of the concept of \emph{abstraction}~\cite{ZeySar21}.
This allows spotting which parts of a given domain are actually relevant for rising the unsatisfiability of the problem.
We plan to explore formal relations to these approaches or to study potential combinations with some of them.
%
\section{Conclusions}
\label{sec:conclusions}

We have introduced the notion of explanation of a model of a logic program as some kind of (acyclic) labelled graph we called \emph{support graph}.
We have have defined justified models as those that have at least one explanation and proved that all stable models are justified, whereas the opposite does not hold, at least for disjunctive programs.
We also provided a meta-programming encoding in ASP that generates the explanations of a given stable model.
We formally proved a one-to-one correspondence between the answer sets of the encoding and the explanations of the original stable model.
Since this encoding constitutes the basis of the tool \xclingo~2.0, we provide in this way a formal proof of correctness for this system.
A system description of the tool is left for a forthcoming document.
Future work includes the comparison to other approaches, the explanation of unsatisfiable programs and the minimisation or even the specification of preferences among explanations.
%
%

\bibliographystyle{plainnat}
\bibliography{refs}
\end{document}